\documentclass[12pt]{iopart}

\expandafter\let\csname equation*\endcsname\relax
\expandafter\let\csname endequation*\endcsname\relax
\usepackage{amsmath}
\usepackage{graphicx}
\usepackage{epstopdf}
\usepackage{multirow}
\usepackage{subfigure}

\usepackage{iopams,cite}
\input xy
\xyoption{all}
\usepackage{amsthm}
\usepackage{amsfonts,amsmath}
\usepackage{lscape}

\newtheorem{den}{Definition}
\newtheorem{axm}{Axioms}

\newtheorem{prn}{Proposition}
\newtheorem{lea}{Lemma}
\newtheorem{coy}{Corollary}
\newtheorem{rek}{Remark}

\def\T{\mathcal T}

\def\B{\mathcal B}
\def\C{\mathcal C}
\def\D{\mathcal D}
\def\E{\varepsilon}
\def\H{\mathcal H}
\def\OO{\mathcal O}

\def\J{\mathcal J}
\def\L{\mathcal L}

\def\U{\mathcal U}

\def\pd{\phantom{\dagger}}

\begin{document}
\title[The Yang-Baxter  paradox ]{The Yang-Baxter paradox}
\author{Jon Links}
\address{School of Mathematics and Physics, The University of Queensland, Brisbane, QLD 4072, Australia}
\begin{abstract}
Consider the statement ``Every Yang-Baxter integrable system is defined to be exactly-solvable''. To formalise this statement, definitions and axioms are introduced. Then, using a specific Yang-Baxter integrable bosonic system, it is shown that a paradox emerges. A generalisation for completely integrable bosonic systems is also developed.  
\end{abstract}

\section{Introduction}

Integrable quantum systems  provide many important insights into, and benchmarks for, many important physical phenomena, e.g. see \cite{bf}. These accomplishments in the application of integrable models predominantly stem from formulation  through the Yang-Baxter equation \cite{yang,ba1,ba2}, and associated exact solutions through Bethe Ansatz techniques \cite{bethe}.  Despite this success, there is no consensus on an adequate definition for what constitutes integrability in quantum systems \cite{cm,erbe,l}.  One of the difficulties in this regard relates to the notion of a 
``maximal'' set of mutually commuting operators. This is because every self-adjoint operator acting on a finite-dimensional vector space of dimension $\D$ commutes with $\D$ idempotent operators that project onto the elements of a basis of eigenstates. Likewise, the notion of an exact solution is imprecise. While the Bethe Ansatz is universally recognised as a extremely powerful technique, it is applied in many guises besides the original co-ordinate formulation of Bethe \cite{bethe}. These include algebraic \cite{tf}, analytic \cite{resh}, functional \cite{holger}, thermodynamic \cite{taka}, off-diagonal \cite{cysw},  double-row transfer matrix constructions \cite{sklyanin}, and by using separation of variables \cite{s}.       
Moreover there are other techniques available to yield exact solutions, such as  the Jordan-Wigner transform \cite{lsm}, those used in Kitaev-type models \cite{kitaev,yu}, and those used in Rabi-type models \cite{braak,gldb,moroz,yao}. These somewhat confuse attempts to provide an unambiguous definition for what constitutes exact-solvability, and to identify its relationship to integrability. The (fermionic) Hubbard model serves as a classic example of how the defiant the challenges can be. Using the co-ordinate Bethe Ansatz, an exact solution was derived in 1968 \cite{lw}. It took until 1986 for the first steps towards establishing Yang-Baxter integrability were completed, with the construction of a transfer matrix \cite{shastry}. A clear understanding of the solution for the Yang-Baxter equation associated with the transfer matrix was delivered in 1995 \cite{usw}. Then, successful application of the algebraic Bethe Ansatz was achieved in 1998 \cite{mr}. This was followed, in 2001, by the construction of the ladder operator to generate the conserved operators \cite{lzmg}.     
By contrast, a similar program for the Rabi model has so far not yielded any signs of Yang-Baxter integrability \cite{bz}.  

In the other direction, the title of Rodney Baxter's famous book \cite{bax} reminds that there is also a distinction to be made between what is {\it exactly-solved}, and what is not yet exactly-solved but possibly {\it exactly-solvable}.  It appears plausible
to suggest that {\it all} Yang-Baxter integrable systems are, in principle, exactly-solvable.
Some solutions are known. Others are unknown, and these serve as open problems for future research. 
However, the
results presented below counter this perspective. 
For classical integrable systems that are generally analysed by quadratures, it has been remarked that ``integrability in classical mechanics does not require the solvability of the quadratures'' \cite{erbe}.    
Likewise, it is argued here that Yang-Baxter integrability in quantum systems does not require exact-solvability. Without recourse to a concrete definition for what constitutes an exact solution, the result is instead deduced through inconsistencies between certain axioms and conclusions.

In Sect. 2 the mathematical framework and notational conventions are established. This includes a brief discussion on the context of the problem via the one-dimensional Bose-Hubbard model, widely considered to not be exactly-solvable. In Sect. 3, a specific Yang-Baxter integrable system is introduced.  Then, in Sect. 4, the main results are derived. This entails two instances where a set of axioms is prescribed to characterise exact-solvability. In both cases, the example from Sect. 3 is used to  show that defining solvability through integrability is inconsistent with the axioms. Concluding discussion is provided in Sect. 5.

\section{Preliminaries}

All Hamiltonians below are formulated through sets of  canonical boson operators $\B = \{b^{\pd}_l,\,b_l^\dagger: l=1,...,\L<\infty \}$ satisfying the commutation relations
\begin{eqnarray*}
&[ b_k^{\pd},\,b_l^\dagger]=\delta_{kl} I, \\
&[ b_k^{\pd},\,b^{\phantom{\dagger}}_l]=[b_k^\dagger,\,b_l^\dagger]=0,
\end{eqnarray*}
where $I$ denotes the identity operator. Throughout, the field is ${\mathbb C}$.
\begin{den} \label{bosham}
A self-adjoint operator $\H$ expressed with non-trivial, polynomial-dependence in the operators
from the universal enveloping algebra of  $\B=\{b^{\pd}_l,\,b_l^\dagger: l=1,...,\L<\infty\}$ is said to be a  {\rm bosonic Hamiltonian}. If $\H$ cannot be expressed in terms of the universal enveloping algebra of a proper subset of $\B,$ then $\H$ is said to have $\L$ {\rm degrees of freedom}.
\end{den}

\begin{den} \label{conserved}
A linear operator $K$ is said to be a conserved operator of a bosonic Hamiltonian $\H$ if $K$ has non-trivial, polynomial-dependence on some of the operators from the universal enveloping algebra of $\B=\{b_l^{\pd},\,b_l^\dagger: l=1,...,\L<\infty \}$ and satisfies 
\begin{eqnarray*}
[\H,\,K]=0. 
\end{eqnarray*} 
\end{den}

\begin{rek}
By defining conserved operators in terms of {\rm polynomial} functions, Definition \ref{conserved} excludes the possibility of identifying projection operators as conserved operators. For example, the projection $P_0$ onto the vacuum $|0\rangle$ is formally given by
\begin{eqnarray*}
P_0=|0\rangle\langle 0|= \prod_{k=1}^\infty \prod_{l=1}^{\L}
(I-k^{-1}b_l^\dagger b^{\phantom{\dagger}}_l)
\end{eqnarray*}
and is consequently not considered to be a conserved operator under Definition \ref{conserved}.
\end{rek}

\begin{den} \label{tpn}
The total number operator $N$ is defined as 
\begin{eqnarray*}
N=\sum_{l=1}^{\L} b_l^\dagger b^{\pd}_l. 
\end{eqnarray*}
If a bosonic Hamiltonian $\H$ satisfies 
\begin{eqnarray*}
[\H,\,N]=0,
\end{eqnarray*}
then $\H$ is said to {\rm conserve total particle number}.

\end{den}

\begin{rek}
A bosonic Hamiltonian $\H$ acts on an infinite-dimensional Fock space. If $\H$ conserves total particle number, then $\H$ 
admits a direct sum decomposition 
\begin{eqnarray}
\H = \bigoplus_{n=0}^{\infty}  \H(n)
\label{decomp}
\end{eqnarray}
where $\H(n)$ acts on a  space $V(n)$ of finite dimension 
\begin{eqnarray*}
{\rm dim}(V(n))=\frac{(n+\L-1)!}{ n!  (\L-1)!},
\end{eqnarray*}
and $N$ acts on $V(n)$ as an $n$-fold multiple of the identity operator.
\end{rek}
 \begin{den}
A set of linear operators $\{\OO_1,\,...,\OO_q\}$ is said to be {\rm functionally-dependent} if there exists a non-constant polynomial $f$ such that  
$f(\OO_1,...,\OO_q)=0$. If a set of linear operators is not functionally-dependent, then it is said to be {\rm functionally-independent}.  
\end{den}

 \begin{den} \label{ci}
A bosonic Hamiltonian $\H$ with $\L$ degrees of freedom is said to be {\rm completely integrable} if there exists a functionally-independent set of conserved operators 
\begin{eqnarray*}
\{   \H=K_1,\,K_2,\,K_3,...,K_{\L}\}
\end{eqnarray*}
satisfying
\begin{eqnarray*}
[K_j,\,K_l]=0, \qquad \forall \,1\leq j,\,l\leq \L .
\end{eqnarray*} 
\end{den}

\begin{rek}
The definition above does not make any reference to dimensionality, nor to locality properties of the conserved operators. Such locality properties are important in the calculation of some physical aspects, such as thermodynamics in one-dimensional systems  e.g. \cite{ef}. But they are not relevant here. Moreover, it bears mentioning that not all applications of Yang-Baxter integrable systems refer to one-dimensional lattice models. Known examples that do not fall into that category include Bloch electrons \cite{wz}, $q$-oscillators 
\cite{sergeev}, central spin models \cite{g,ng,s}, and superconductor pairing models \cite{cdv,dilsz,lil,skrypnyk09,vp,zlmg}.  
 \end{rek}

\begin{den} \label{ybint}
A bosonic Hamiltonian $\H$ is said to be {\rm Yang-Baxter integrable} if there exists an operator $t(u),\,u\in \,{\mathbb C}$, known as a {\rm transfer matrix}, that has non-trivial, polynomial-dependence on some of the operators from the universal enveloping algebra of $\B=\{b_l^{\pd},\,b_l^\dagger: l=1,...,\L<\infty\}$ and satisfies 
\begin{eqnarray*}
&[\H,\,t(u)]=0, \\
&[t(u), \, t(v)]=0, \qquad \forall \, u,\, v \in\, {\mathbb C}.
\end{eqnarray*}  
\end{den}

\begin{rek}
Historically, a transfer matrix draws its name from models in lattice statistical mechanics. The above definition is, admittedly, vague in that it does not refer to a Yang-Baxter equation nor its relation to the transfer matrix. This is deliberate because there are numerous approaches that can be employed to construct commuting transfer matrices beyond the most familiar cases with periodic boundary conditions. These include, but are not limited to, the use of classical Yang-Baxter equations \cite{bd,jurco}, dynamical Yang-Baxter equation \cite{ev}, star-triangle relations \cite{helen}, tetrahedron equation \cite{bs}, and through quantum determinants and Dunkl operators \cite{bghp}. They may incorporate twisted boundary conditions \cite{nw}, open boundary conditions \cite{sklyanin}, braided boundary conditions \cite{lf}, or infinite chains \cite{omar}.  However, these variants are not pertinent below because only one representative of Yang-Baxter integrability will be required to develop the arguments. 
\end{rek}

\subsection{The one-dimensional Bose-Hubbard model}
Next, to provide some background context, a brief discussion is given to the one-dimensional Bose-Hubbard model. Much of the detail is adapted from \cite{ol}. 

The system satisfies Definition \ref{bosham}, with the explicit Hamiltonian
\begin{eqnarray}
H_{BH}= -t (b^\dagger_1 b^{\pd}_{\L} + b^\dagger_{\L} b^{\pd}_1)   -t\sum_{j=1}^{\L-1} (b^\dagger_j b^{\pd}_{j+1} + b^\dagger_{j+1} b^{\pd}_j) 
+\U\sum_{j=1}^{\L} b^\dagger_j b^\dagger_j b^{\pd}_j b^{\pd}_j
\label{bh}
\end{eqnarray} 
from which it can be verified that Definition \ref{tpn} is met such that Eq. (\ref{decomp}) is satisfied. It follows that for $\L=2$, the so-called Bose-Hubbard {\it dimer} model, the system is completely integrable according to Definition \ref{ci}, with the set of conserved operators $\{H_{BH},\, N\}$. For this case it is known that the system is Yang-Baxter integrable via a transfer matrix associated with either the classical Yang-Baxter equation \cite{eks}
or  the quantum Yang-Baxter equation \cite{esks,esse}, and is exactly-solvable. See also \cite{lh}. The Hamiltonian (\ref{bh}) is again exactly-solvable in the limit $\L\rightarrow\infty$ \cite{ll,jim}. For other values of $\L$ with $n=1$ the system can be solved by discrete Fourier transform, and for $n=2$ by using centre-of-mass coordinates. But attempts to obtain solutions for higher values of $n$ have not succeeded \cite{ch}. The expectation that the model is not generally solvable is consistent with the characterisation of chaotic behaviour found when $\L=3$ \cite{fp}.

Suppose it is accepted that (\ref{bh}) is not exactly-solvable. Then any definition for that characterisation needs to be able to identify, for some $\L$ and $n$, that the action of (\ref{bh}) on $V(n)$ is not exactly-solvable. 
That is to say, for bosonic Hamiltonians that conserve total particle number it is meaningful to refer to not exactly-solvable linear operators acting on finite-dimensional spaces, precisely those spaces within the decomposition (\ref{decomp}).          
If this was not the case, and each linear operator obtained by restricting (\ref{bh}) to $V(n)$ was exactly-solvable, then the  exact solution on the entire space would simply be obtained from the union of the exact solutions for each subspace.

\section{A class of bosonic Yang-Baxter  integrable systems}

In this section a construction is described, via an explicit transfer matrix, for a class of bosonic Hamiltonians that are Yang-Baxter integrable. For notational convenience, the following sets of canonical boson operators are introduced, indexed by multiple labels:
\begin{eqnarray}
\B_j= \{   a^{\pd}_{(j,\mu)},\,a_{(j,\mu)}^\dagger, \, b^{\pd}_{(j,\mu)},\,b_{(j,\mu)}^\dagger: \, \mu=1,...,m_j  \}, \nonumber \\
\B=\bigcup_{j=1}^L \B_j \label{bee}
\end{eqnarray}
whereby 
\begin{eqnarray}
2\sum_{j=1}^{L} m_j=\L. 
\label{sum}
\end{eqnarray}
Set the following notations for number operators:
\begin{eqnarray*}
N_j=\sum_{\mu=1}^{m_j}  \left( a_{(j,\mu)}^{\dagger} a^{\pd}_{(j,\mu)} + b_{(j,\mu)}^{\dagger} b^{\pd}_{(j,\mu)} \right) , \\  
N_a= \sum_{j=1}^L    \sum_{\mu=1}^{m_j}  a_{(j,\mu)}^{\dagger} a^{\pd}_{(j,\mu)} ,  \\ 
N_b= \sum_{j=1}^L    \sum_{\mu=1}^{m_j}  b_{(j,\mu)}^{\dagger} b^{\pd}_{(j,\mu)} ,  \\ 
N=N_a+N_b= \sum_{j=1}^L N_j .
\end{eqnarray*}
Let $\{ \E_j: \,j=1,...,L\}$ denote a set of arbitrary, pairwise distinct, real parameters.
\begin{prn} \label{ybi}
The bosonic Hamiltonian 
\begin{eqnarray}
H=U(N_a-N_b)^2 + \sum_{j=1}^L \E_j   \sum_{\mu=1}^{m_j}  \left(a_{(j,\mu)}^{\dagger} b^{\pd}_{(j,\mu)} + b_{(j,\mu)}^{\dagger} a^{\pd}_{(j,\mu)} \right) ,
\qquad U\in\,{\mathbb R},
\label{ham}
\end{eqnarray}
is Yang-Baxter integrable.
\end{prn}
\begin{proof}
This result is a generalisation of one found in \cite{links}, to which it reduces when $m_j=1$ for all $j=1,...,L=\L/2$. The key elements are outlined below. \\

\noindent
For $u,\,v\in{\mathbb C}$, let $r(u,v)\in\,{\rm End}({\mathbb C}^2\otimes {\mathbb C}^2)$.
A form of classical Yang-Baxter equation reads  \cite{at,maillet,torr}
\begin{eqnarray}
\left[r_{12}(u,v),\,r_{23}(v,w)\right]-\left[r_{21}(v,u),\,r_{13}(u,w)\right]+\left[r_{13}(u,w),\,r_{23}(v,w)\right] =0
\label{cybe}
\end{eqnarray}
where the subscripts refer to the embedding of $r(u,v)$  in ${\rm End}({\mathbb C}^2\otimes {\mathbb C}^2 \otimes 
{\mathbb C}^2)$.  
A solution for (\ref{cybe}) is given by \cite{skrypnyk09} 
\begin{eqnarray} 
r(u,v)&= \left(\begin{array}{ccccc} 
\displaystyle\frac{2u^2}{u^2-v^2} &0&|&0&0 \\
0&0 &|&\displaystyle\frac{2uv}{u^2-v^2}&0 \\ 
-&-&~&-&- \\ 
0&\displaystyle\frac{2uv}{u^2-v^2}&|&0 &0 \\
0&0&|&0&\displaystyle\frac{2u^2}{u^2-v^2}  \end{array} \right).
\label{rmat} 
\end{eqnarray}
It may be compactly expressed as
\begin{eqnarray*}
r(u,v)=\left(\frac{u}{u-v}I\otimes I+\frac{u}{u+v}\sigma^z_1 \sigma^z_2\right)P 
\end{eqnarray*}
where $P$ is the permutation operator and $\sigma^z={\rm diag}(1,-1)$. The solution (\ref{rmat}) satisfies \cite{links}
\begin{eqnarray*}
&  [\J_2(v),\,r_{12}(u,v)]-[\J_1(u),\,r_{21}(v,u) ]=0,
\end{eqnarray*}
where 
\begin{eqnarray*}
\J(u)=\left( \begin{array}{cc}
\alpha & u \beta \\
u\beta & -\alpha
\end{array} 
\right), \qquad \alpha,\,\beta \in\,{\mathbb C}.
\end{eqnarray*}
Next, consider the $L$-fold tensor product  of universal enveloping algebras for $sl(2)$ with generators $\{S^z_j,\,S^+_j,\,S_j^-:\,j=1,...,L\}$, satisfying the commutation relations
\begin{eqnarray}
[S_j^z,\,S_k^{\pm}]=\pm \delta_{jk}S_k^{\pm}, \qquad [S_j^+,\, S_k^-]=2\delta_{jk}S_k^z,
\label{cr}
\end{eqnarray}
and identify the Casimir invariants as $C_j=2(S_j^z)^2+S_j^+S_j^-+S_j^-S_j^+$. 
Setting
\begin{eqnarray*}
T^1_1(u)&=\alpha I+\sum_{j=1}^L\frac{u^2}{u^2-\E_j^2}(I+2S_j^z),   \\  
T^1_2(u)&=u\beta I+\sum_{j=1}^L\frac{2u\E_j}{u^2-\E_j^2}S_j^+,    \\
T^2_1(u)&=u\beta I+\sum_{j=1}^L\frac{2u\E_j}{u^2-\E_j^2}S_j^-,  \\
T^2_2(u)&=-\alpha I+\sum_{j=1}^L\frac{u^2}{u^2-\E_j^2}(I-2S_j^z).
\end{eqnarray*}
and
\begin{eqnarray*}
\T(u)&=\left(
\begin{array}{ccc} 
T^1_1(u) & T^2_1(u) \\
T^1_2(u) & T^2_2(u)  
\end{array} 
\right),
\end{eqnarray*} 
it can be shown that 
\begin{eqnarray*}
\left[ \T_1(u),\,\T_2(v)\right]=\left[ \T_2(v),\,r_{12}(u,v)   \right]-\left[\T_1(u),\,r_{21}(v,u)  \right].
\end{eqnarray*}
Then the algebraic transfer matrix
\begin{eqnarray}
\tau(u)={\rm tr}\left( (\T(u))^2 \right)=\sum_{j,k=1}^2T^j_k(u)T^k_j(u)
\label{tm}
\end{eqnarray}
satisfies 
\begin{eqnarray*}
\left[\tau(u),\,\tau(v)\right]=0 \qquad \forall\,u,\,v\in\,{\mathbb C}.
\end{eqnarray*}
In detail, 
\begin{eqnarray*}
\tau(u)
= 2(\alpha^2+\beta^2u^2)I  &+ 2u^4\left(\sum_{j=1}^L\frac{1}{u^2-\E_j^2}\right)^2I    
\\
&+4u^2\sum_{j=1}^L\frac{\E_j^2}{(u^2-\E_j^2)^2}C_j+4\sum_{j=1}^L\frac{u^2}{u^2-\E_j^2}\T_j
\end{eqnarray*}
where 
\begin{eqnarray}
\T_j=2(S_j^z)^2+2\alpha S_j^z+\beta \E_j (S_j^++S_j^-)+\sum_{k\neq j} \theta_{jk}, \label{ttees}\\ 
\theta_{jk}=\frac{4\E^2_j}{\E_j^2-\E_k^2}S_j^zS_k^z+\frac{2\E_j\E_k}{\E_j^2-\E_k^2}(S_j^+S_k^-+S_j^-S_k^+), 
\qquad j\neq k.
\nonumber
\end{eqnarray}
By construction, the operators (\ref{ttees}) satisfy
\begin{eqnarray}
\left[\T_j,\,\T_k\right]=0, \qquad\qquad 1\leq j, k\leq L.
\label{commute} 
\end{eqnarray}
A direct verification of Eq. (\ref{commute}) is provided in the Appendix. The operators (\ref{ttees}) are recognised as generalisations of Gaudin operators \cite{cdv,g,ng,s}. \\

\noindent
A representation $\pi$ for the spin operators, in terms of the boson operators of (\ref{bee}), is afforded by the Jordan-Schwinger map
\begin{eqnarray*}
&\pi(S^z_j)=\frac{1}{2}\sum_{\mu=1}^{m_j}(a^\dagger_{(j,\mu)}a^{\phantom{\dagger}}_{(j,\mu)}
-b^\dagger_{(j,\mu)} b^{\phantom{\dagger}}_{(j,\mu)}), \\
&\pi(S^+_j)=\sum_{\mu=1}^{m_j} a^\dagger_{(j,\mu)}b^{\phantom{\dagger}}_{(j,\mu)},  \\
&\pi(S^-_j)=\sum_{\mu=1}^{m_j} b^\dagger_{(j,\mu)}a^{\phantom{\dagger}}_{(j,\mu)},
\end{eqnarray*}
and extends to the entire $L$-fold universal enveloping algebra as an algebra homomorphism. Finally, set $\alpha=0, \, 2\beta=U^{-1} $, and 
\begin{eqnarray*}
&t(u)=\pi(\tau(u)).
\end{eqnarray*}
Then
\begin{eqnarray*}
H= 2U\sum_{j=1}^L  \pi(\T_j)
\end{eqnarray*}
and satisfies $[H,\,t(u)]=0$, establishing the Yang-Baxter integrability of (\ref{ham}).
\end{proof}
\begin{rek}
When $L=\L/2=1$, the Hamiltonian (\ref{ham}) is equivalent to the dimer model limit of (\ref{bh}). Thus (\ref{ham}) can be viewed as an integrable, many-degrees of freedom, generalisation of the Bose-Hubbard dimer model, while (\ref{bh}) is a non-integrable generalisation of the same model.  
\end{rek}

\begin{coy} \label{cons}
There exists a set of conserved operators $\C= \{   H=K_1,\,K_2,\,K_3,...,K_{2L}\}$ for $H$, with elements given by 
\begin{eqnarray*}
&K_j= \pi(\T_j), \qquad &j=2,..,L,  \\
&K_{L+j}= \pi(C_j), \qquad &j=1,...,L. 
\end{eqnarray*}   
\end{coy}

\begin{lea}
If $m_j=1$ for all $j=1,...,L=\L/2$, then $\C$ is functionally-independent provided $\E_j\neq 0$ for all $j=1,...,L$.
\end{lea}
\begin{proof}
Each of the elements of $\C$ is a linear function in $U$. It is seen that the limiting operators 
\begin{eqnarray*}
&\bar{K}_1=\lim _{U\rightarrow 0 } K_1= \sum_{j=1}^L \E_j   \left(a_{(j,1)}^{\dagger} b^{\pd}_{(j,1)} + b_{(j,1)}^{\dagger} a^{\pd}_{(j,1)} \right) ,\qquad & \\
&\bar{K}_j=\lim _{U\rightarrow 0 } K_j= \E_j   \left(a_{(j,1)}^{\dagger} b^{\pd}_{(j,1)} + b_{(j,1)}^{\dagger} a^{\pd}_{(j,1)} \right)  , 
\qquad &j=2,...,L,\\
&\bar{K}_{L+j}=\lim _{U\rightarrow 0 } K_{L+j}= \frac{1}{2}N_j(N_j+2I)  ,
\qquad &j=2,...,L, \\  
\end{eqnarray*}
form a functionally-independent set 
$
\bar{\C}
=\{\bar{K}_1,\,\bar{K}_2,...,\bar{K}_{2L}\}
$. 
Now assume there exists a polynomial $f$ such that  
\begin{eqnarray*}
f({K}_1,...,{K}_{2L})=0.
\end{eqnarray*}
Let $r\in\,{\mathbb R}$ denote the smallest number such that  
\begin{eqnarray*}
\bar{f}=\lim_{U\rightarrow 0} U^r f 
\end{eqnarray*}
is a non-zero polynomial. Then
\begin{eqnarray*}
0&=f({K}_1,...,{K}_{2L}) \\
&=\lim_{U\rightarrow 0} U^r f({K}_1,...,{K}_{2L}) \\
&=\bar{f}(\bar{K}_1,...,\bar{K}_{2L})
\end{eqnarray*}
which is a contradiction since $\bar{\C}$ is functionally-independent. Hence,  ${\C}$ must be 
functionally-independent.
\end{proof}
\begin{coy} \label{coyci}
If $m_j=1$ for all $j=1,...,L=\L/2$ and $\E_j\neq 0$ for all $j=1,...,L=\L/2$, then (\ref{ham}) is completely integrable according to Definition \ref{ci}.
\end{coy}
\begin{rek}
It is not the case that all conserved operators of a Yang-Baxter integrable system are necessarily obtained from the transfer matrix. For example, the class of models in \cite{ytfl} have only two independent conserved operators originating from the transfer matrix. However those models are completely integrable in the sense of Definition \ref{ci}. If not all $m_j=1$, Corollary \ref{coyci} does not assert that (\ref{ham}) is not completely integrable. However if it is completely integrable, alternative methods are required to establish this property. 
\end{rek}

\section{Paradoxical statements on exact-solvability} \label{paradoxes}

The main results can now be formulated. 

\subsection{Yang-Baxter integrable systems}

\begin{axm}
\label{ax1}
A definition for exactly-solvable, self-adjoint, linear operators acting on  complex  vector spaces satisfies the following axioms:
\begin{itemize}
\item[(i)] The set of operators acting on finite-dimensional spaces that are not exactly-solvable is non-empty;
\item[(ii)] Given an arbitrary operator $Y$ and an operator $Z$ that is not exactly-solvable, the direct sum $Y\oplus Z$ is not exactly-solvable.
\end{itemize}
\end{axm}
\begin{rek}
The above itemised list is, in no way, intended to be exhaustive of suitable axioms to impose on a definition for exact-solvability. This is a minimal list for the purpose of establishing the results. Condition (i) is essential, otherwise there would be nothing to discuss. Without item (ii) it would be possible for an operator to be exactly-solvable when its action on a subspace spanned by a subset of eigenstates is not. Such a scenario is difficult to justify.  
\end{rek}
\noindent 
Let $A$ denote a self-adjoint, linear operator acting on a complex vector space of dimension $\D<\infty$, with matrix elements given by 
$\{A_{jk}: j,k=1,...,\D\}$. Let $X_{q(k,\nu)}$ denote the matrix elements of a unitary operator $X$ that diagonalises $A$, viz.
\begin{eqnarray*}
\sum_{p,q=1}^{\mathcal D} X^*_{(j,\mu)p}A_{pq} X_{q(k,\nu)} = \varepsilon_j \delta_{jk} \delta_{\mu\nu}, 
\end{eqnarray*}
with $\{\varepsilon_j:j=1,...,L\}$ the spectrum of $A$, each $\E_j$ occurring with multiplicity $m_j$, such that 
\begin{eqnarray*}
\sum_{j=1}^{L} m_j=\D. 
\end{eqnarray*}
Set $\overline{\B}=\{c^{\pd}_l,\,c_l^\dagger,\,  d^{\pd}_l,\,d_l^\dagger: l=1,...,\D\}$ so $\L=2\D$ and  Eq. (\ref{sum}) 
holds. Define 
\begin{eqnarray}
{\mathbb H}=U(N_c-N_d)^2-UI+ 
\sum_{j,k=1}^{\D} A_{jk}(c^\dagger_k d^{\pd}_j + d_k^\dagger c^{\pd}_j).
\label{newham}
\end{eqnarray}
where
\begin{eqnarray*}
N_c=     \sum_{j=1}^{\D}  c_{j}^{\dagger} c^{\pd}_{j} ,  \qquad 
N_d=     \sum_{j=1}^{\D}  d_{j}^{\dagger} d^{\pd}_{j}.
\end{eqnarray*}
\begin{prn} \label{woot}
The Hamiltonian (\ref{newham}) is Yang-Baxter integrable.
\end{prn}  

\begin{proof}
Introducing operators 
\begin{eqnarray*}
a_{(k,\nu)}=\sum_{j=1}^\D X_{j(k,\nu)} c_j, \qquad
b_{(k,\nu)}=\sum_{j=1}^\D X_{j(k,\nu)} d_j,
\end{eqnarray*}
leads to $N_c=N_a$, $N_d=N_b$ and ${\mathbb H}+UI=H$ with $H$ given by (\ref{ham}). Thus, (\ref{newham}) is Yang-Baxter integrable by Proposition \ref{ybi}.
\end{proof}
\begin{rek}
The class of Hamiltonians with the form (\ref{newham}) generalises a class of Hamiltonians introduced in 
\cite{ytfl}. The classes coincide when $A$ has rank 1.  
\end{rek}
\begin{prn} \label{ns}
Suppose that a specific $A$ has been deemed to not be exactly-solvable, according a definition satisfying Axioms \ref{ax1}. Then the corresponding Hamiltonian (\ref{newham}) is not exactly-solvable.
\end{prn}  
\begin{proof}
The Hamiltonian (\ref{newham}) commutes with the number operator 
\begin{align*}
N=N_a+N_b=N_c+N_d
\end{align*}
and therefore (\ref{newham}) can be block-diagonalised as in (\ref{decomp}). 
Now, ${\mathbb H}(1)$ acts on a space of dimension $2\D$, and can be represented as a tensor product of linear operators
\begin{eqnarray*}
&{\mathbb H}(1)\cong \left(
\begin{array}{ccc}
0 & | & A \\
- & & - \\
A & | & 0 
\end{array}
\right) \\
&\phantom{{\mathbb H}(1)}\cong \left(
\begin{array}{ccc}
A & | & 0 \\
- & & - \\
0 & | & -A 
\end{array}
\right). 
\end{eqnarray*} 
Since $A$ is not exactly-solvable, it then follows that ${\mathbb H}(1)$ is not exactly-solvable, and so ${\mathbb H}$ is not exactly-solvable, according to Axioms \ref{ax1}(ii).  
\end{proof}
\begin{coy}
The definition \\

\noindent
{\rm ``A bosonic Hamiltonian (see Definition \ref{bosham}) is {\it exactly-solvable} if it is Yang-Baxter integrable (see Definition \ref{ybint})''} \\

\noindent
violates Propositions \ref{woot} and \ref{ns}, and is therefore inconsistent with Axioms \ref{ax1}. 
\end{coy}

\subsection{Completely integrable systems}

\begin{axm}
\label{ax2}
A definition for exactly-solvable, self-adjoint, linear operators acting on  complex vector spaces satisfies the following axioms: 
\begin{itemize}
\item[(i)] The set of operators acting on finite-dimensional spaces that are not exactly-solvable, and have multiplicity-free spectrum, is non-empty;
\item[(ii)] Given an arbitrary operator $Y$ and an operator $Z$ that is not exactly-solvable, the direct sum $Y\oplus Z$ is not exactly-solvable;
\item[(iii)] Given an operator $Z$ that is not exactly-solvable, then $Z+\lambda I$, $\lambda\in \,{\mathbb R}$, where $I$ is the identity operator, is not exactly-solvable. 
\end{itemize}
\end{axm}

\begin{rek}
The  added restriction on item (i) is minimal. Any operator that has degenerate spectrum can have the degeneracy broken by an arbitrarily small perturbation. It is reasonable to impose that operators that are not exactly-solvable remain so, generically,  under arbitrarily small perturbation. Item (iii) seems to be self-evident, but needs to be formally stated.
\end{rek}

\begin{coy} \label{colci}
Choose $\tilde{A}$ to have multiplicity-free spectrum and $\lambda\in\,{\mathbb R}$ such that $A=\tilde{A}+\lambda I$ has positive eigenvalues. Then the  Hamiltonian (\ref{newham}) is completely integrable as a result of Corollary \ref{coyci}, and the identification ${\mathbb H}+UI=H$ from the proof of Proposition \ref{woot}.  If $\tilde{A}$ is deemed to be not exactly-solvable, 
according to a definition satisfying Axioms \ref{ax2}, then $A$ is not exactly-solvable and so (\ref{newham}) is not exactly-solvable by Proposition \ref{ns}. 
\end{coy}  

\begin{coy}
The definition \\

\noindent
{\rm ``A bosonic Hamiltonian (see Definition \ref{bosham}) is {\it exactly-solvable} if it is completely integrable (see Definition \ref{ci})''} \\

\noindent
is inconsistent with Axioms \ref{ax2}. 
\end{coy}
\section{Discussion}
Many paradoxes arise as a consequence of self-reference. Some are popularly known, such as the Liar paradox and the Barber paradox. Others, including Curry's paradox and the Grelling-Nelson paradox, have had important impact in developing formal logic. The most prominent examples invoking self-reference in mathematics are arguably those associated with Russell and G\"odel.  See \cite{logic}, for example.

The paradoxes of Sect. \ref{paradoxes} also arise through self-reference. The Hamiltonian (\ref{newham}) is defined in terms of an operator $A$ that is assumed to possess a particular property, i.e. $A$ is not exactly-solvable. However, the class of integrable systems to which (\ref{newham}) belongs is used to define the {\it negation} of that same property. The surprise is not so much that a paradox occurs. It is, rather, the existence of the Yang-Baxter integrable system (\ref{newham}) that is sufficiently general to expedite such a self-referential construction. 

In an attempt to resolve the paradoxes, consideration can be given to changing relevant axioms and definitions. For example, a simple resolution is obtained by declaring that all self-adjoint operators on finite-dimensional spaces are exactly-solvable. It is not an unreasonable position to take, because the spectrum of each operator is determined by its characteristic polynomial. But it does not help to understand why some systems are amenable to a Bethe Ansatz solution and others are not. Another route might be to revise Definition \ref{ci}. Rather than  define complete integrability in terms of the existence of conserved operators, it could be weakened to only include cases where the explicit form of the conserved operators is known. This feature does not currently apply to (\ref{newham}) when $A$ is not exactly-solvable. Expressing the conserved operators of Corollary \ref{cons}  in terms of the generating set $\overline{\B}$ that defines (\ref{newham}) leads to expressions with explicit dependence on the unknown quantities $X_{q(k,\nu)}$ and 
$\varepsilon_j$. However the drawback here is that this system, now classified as non-integrable, has exactly the same energy-level statistics as an integrable system.  Integrable systems are generally expected to display a signature Poisson distribution for the energy gaps, as has been observed in bosonic systems \cite{dlmz} analogous to those discussed above. By defining integrability in a basis-dependent manner, it would subsequently render all studies relating integrability to energy-level statistics as meaningless. 

The above results were formulated around a solution of the classical Yang-Baxter equation (\ref{cybe}). This is not of a standard form, with the corresponding solution (\ref{rmat}) not possessing the skew-symmetry property $r_{12}(u,v)=-r_{21}(v,u)$. The construction of the transfer matrix from this solution can be understood as the quasi-classical limit of a double-row transfer matrix built around reflection equations \cite{lil,skrypnyk07}, rather than the conventional Yang-Baxter equation. This observation suggests that it may be possible to lift the construction presented here into that double-row transfer matrix setting, to extend the paradox to a wider class of Yang-Baxter integrable systems.

\ack{This work was funded by the Australian Research Council through Discovery
Project DP200101339. The author acknowledges the traditional
owners of the land on which The University of Queensland operates, the Turrbal and Jagera people.}

\section*{Appendix}

Here it is shown that Eq. (\ref{commute}) holds, viz. 
\begin{eqnarray*}
\left[\T_j,\,\T_k\right]=0, \qquad\qquad 1\leq j, k\leq L,
\end{eqnarray*}
by direct use of the commutation relations (\ref{cr}). Recall that 
\begin{eqnarray*}
\T_j=2(S_j^z)^2+2\alpha S_j^z+\beta \E_j (S_j^++S_j^-)+\sum_{k\neq j}^L \theta_{jk}, \\ 
\theta_{jk}=\frac{4\E^2_j}{\E_j^2-\E_k^2}S_j^zS_k^z+\frac{2\E_j\E_k}{\E_j^2-\E_k^2}(S_j^+S_k^-+S_j^-S_k^+),
\qquad j\neq k.
\end{eqnarray*} 
First note that for $j\neq p$
\begin{eqnarray*}
[\theta_{jk},\,\theta_{pk}]
  &=  \frac{8\E^2_j}{\E_j^2-\E_k^2} \frac{\E_p\E_k}{\E_p^2-\E_k^2} S_j^z ( S_k^+ S_p^-- S_k^-S_p^+)  \\
&\qquad +\frac{8\E^2_p}{\E_p^2-\E_k^2}\frac{\E_j\E_k}{\E_j^2-\E_k^2}  S_p^z (S_j^+S_k^- - S_j^-S_k^+)  \\
&\qquad + \frac{8\E^2_k}{\E_j^2-\E_k^2}  \frac{\E_j\E_p}{\E_p^2-\E_k^2}  S_k^z (S_p^+S_j^- - S_p^- S_j^+  ), 
\end{eqnarray*}
yielding 
\begin{eqnarray*}
[\theta_{pk}, \,\theta_{jk}] &+[\theta_{kj}, \,\theta_{pj}] + [\theta_{jp}, \,\theta_{kp}] =0.
\end{eqnarray*}
Now from 
\begin{eqnarray*}
&\theta_{jk}+\theta_{kj}= 4 S_j^z S_k^z, \\
&[S_k^z S_p^z,\, \theta_{jk}]= 
\frac{2\E_j\E_k}{\E_j^2-\E_k^2}  S_p^z (S_k^+ S_j^- - S_k^- S_j^+) 
=  [S_{j}^z S_p^z, \,\theta_{kj}],   \qquad j\neq p,
\end{eqnarray*}
it follows that
%
\begin{eqnarray*} [\theta_{jk},\,\theta_{kp}]  + [\theta_{jp},\,\theta_{kj} ]  + [\theta_{jp},\,\theta_{kp}]  
&= [\theta_{pk},\,\theta_{jk}]  + [\theta_{kj},\,\theta_{pj} ]  + [\theta_{jp},\,\theta_{kp}]  \\
&=0
\end{eqnarray*}
leading to
\begin{eqnarray*}
\sum_{p\neq j}^L\sum_{l\neq k}^L [\theta_{jp},\,\theta_{kl}] 
&=[\theta_{jk},\,\theta_{kj}]
+\sum_{p\neq j,k}^L[\theta_{jp},\,\theta_{kp}]
+\sum_{p\neq j,k}^L[\theta_{jp},\,\theta_{kj}]
+\sum_{p\neq j,k}^L[\theta_{jk},\,\theta_{kp}]\\
&= 4[ S_j^z S_k^z ,\, \theta_{kj} ]  \\
&=\frac{8\E_j \E_k}{\E_k^2-\E_j^2}  \left( S_j^z ( S_k^+ S_j^- - S_k^- S_j^+)  
+   ( S_j^+ S_k^- - S_j^- S_k^+)  S_k^z   \right)   .
\end{eqnarray*}
Also,
\begin{eqnarray*}
[ S^z_j , \, \theta_{kj}] &=   \frac{2 \E_k \E_j}{\E_k^2-\E_j^2}(S_j^+S_k^- - S_j^-S_k^+  )  \\
&= [ S^z_k , \, \theta_{jk}]  
\end{eqnarray*}
and
\begin{eqnarray*}
&\E_j[S_j^+ +S_j^-   , \, \theta_{kj}   ] &=\E_j \left( \frac{4\E^2_k}{\E_j^2-\E_k^2}
S_k^z(S_j^+ - S_j^-   )+\frac{4\E_k\E_j}{\E_k^2-\E_j^2}S_j^z(S_k^+ - S_k^-) \right) \\
&&=\E_k[S_k^+ +S_k^-   , \, \theta_{jk}   ]   
\end{eqnarray*}

Then, for $j\neq k$,
\begin{eqnarray*}
\left[\T_j,\,\T_k\right]
&=\sum_{l\neq k}^L[2(S_j^z)^2+2\alpha S_j^z+\beta \E_j (S_j^++S_j^-),\, \theta_{kl}] \\
&\qquad+\sum_{p\neq j}^L[\theta_{jp},\, 2(S_k^z)^2+2\alpha S_k^z+\beta \E_k (S_k^++S_k^-)] 
+\sum_{p\neq j}^L\sum_{l\neq k}^L [\theta_{jp},\,\theta_{kl}] \\
&=2S_j^z[S_j^z,\,\theta_{kj}]+2[S_j^z,\,\theta_{kj}]S_j^z  + 2S_k^z[\theta_{jk},\,S_k^z]+2[\theta_{jk},\,S_k^z]S_k^z  \\
&\qquad + \frac{8\E_j \E_k}{\E_k^2-\E_j^2}  \left( S_j^z ( S_k^+ S_j^- - S_k^- S_j^+)  +   ( S_j^+ S_k^- - S_j^- S_k^+)  S_k^z   \right)  \\
&=2(S_j^z-S_k^z)[S_j^z,\,\theta_{kj}]+2[S_j^z,\,\theta_{kj}](S_j^z-S_k^z)    \\
&\qquad + \frac{8\E_j \E_k}{\E_k^2-\E_j^2}  \left( S_j^z ( S_k^+ S_j^- - S_k^- S_j^+)  +   ( S_j^+ S_k^- - S_j^- S_k^+)  S_k^z   \right)  \\
&=\frac{4\E_k\E_j}{\E_k^2-\E_j^2}  \left((S_j^z-S_k^z) ( S_j^+S_k^--S_j^-S_k^+  )  
+ (  S_j^+S_k^--S_j^-S_k^+)  (S_j^z-S_k^z) \right)  \\
&\qquad + \frac{8\E_j \E_k}{\E_k^2-\E_j^2}  \left( S_j^z ( S_k^+ S_j^- - S_k^- S_j^+)  +   ( S_j^+ S_k^- - S_j^- S_k^+)  S_k^z   \right)  \\
&=\frac{4\E_k\E_j}{\E_k^2-\E_j^2} [ \left( S_j^+S_k^--S_j^-S_k^+\right), \, (S_j^z+S_k^z) ]  \\
&=0.
\end{eqnarray*}


\section*{References}
\begin{thebibliography}{99}
\bibitem{helen} H. Au-Yang and J.H.H. Perk, {\it Onsager’s star-triangle equation: Master key to integrability}, Advanced Studies in Pure Mathematics {\bf 19}, 57 (1989).
\bibitem{at} J. Avan and M. Talon, {\it Rational and trigonometric constant non-antisymmetric $r$-matrices}, Phys. Lett. B {\bf 241},  77
(1990).
\bibitem{bf} M.T. Batchelor and A. Foerster, {\it Yang--Baxter integrable models in experiments: from condensed matter to ultracold atoms}, J. Phys. A: Math. Theor. {\bf 49}, 173001 (2016). 
\bibitem{bz}  M.T. Batchelor and H.-Q. Zhou, {\it Integrability versus exact solvability in the quantum Rabi and Dicke models},
Phys. Rev. A {\bf 91}, 053808 (2015).
\bibitem{ba1} R.J. Baxter, {\it Partition function of the eight-vertex lattice model}, Ann. Phys. {\bf 70}, 193 (1972).
\bibitem{ba2} R.J. Baxter, {\it Eight-vertex model in lattice statistics and one-dimensional anisotropic Heisenberg chain. I. Some fundamental eigenvectors}, Ann. Phys. {\bf 76}, 1 (1973).
\bibitem{bax} R.J. Baxter,   {\it Exactly Solved Models in Statistical Mechanics} (London: Academic, 1982).
\bibitem{bs} V.V. Bazhanov and S.M, Sergeev, {\it Tetrahedron equation and hidden structure of quantum groups}, J. Phys. A: Math. Gen. {\bf 39}, 3295  (2006). 
\bibitem{bd} A.A. Belavin and V.G. Drinfel’d, {\it Solutions of the classical Yang–Baxter equation for simple Lie algebras}, Funct. Anal. Appl. {\bf 16}, 159 (1982).
\bibitem{bghp} D. Bernard, M. Gaudin, F.D.M. Haldane, and V. Pasquier,
{\it Yang-Baxter equation in long-range interacting systems}, J. Phys. A: Math. Gen. {\bf 26}, 5219 (1993).
\bibitem{bethe} H. Bethe, {\it Zur Theorie der Metalle. I. Eigenwerte und Eigenfunktionen der linearen Atomkette}, Z. Phys. {\bf 71}, 205 (1931). 
\bibitem{braak} D. Braak, {\it Integrability of the Rabi model}, Phys. Rev. Lett. {\bf 107}, 100401 (2011).
\bibitem{cysw} J. Cao, W.-L. Yang, K. Shi, and Y. Wang, {\it Off-diagonal Bethe Ansatz and exact solution of a topological spin ring}, 
Phys. Rev. Lett. {\bf 111}, 137201 (2013).
\bibitem{cm} J.-S. Caux and J. Mossel, {\it Remarks on the notion of quantum integrability}, J. Stat. Mech. P02023 (2011).  
\bibitem{ch} T.C. Choy and F.D.M. Haldane, {\it Failure of the 
Bethe-Ansatz solutions of generalisations of the Hubbard chain to arbitrary permutation symmetry}, Phys. Lett. {\bf 90}, 83 (1982).
\bibitem{cdv} P.W. Claeys, S. De Baerdemacker, and D. Van Neck, {\it Read-Green resonances in a topological superconductor coupled to a bath}, 
Phys. Rev. B {\bf 93}, 220503(R) (2016).
\bibitem{logic} J.N. Crossley, C.J. Ash, C.J. Brickhill, J.C. Stillwell, and N.H. Williams, {\it What is Mathematical Logic?} (Oxford University Press, 1972).
\bibitem{dlmz} S.R. Dahmen, J. Links, R.H. McKenzie, and
H.-Q. Zhou, {\it Energy level statistics for models of
coupled single-mode Bose-Einstein condensates}, J. Stat. Mech. P10019 (2004). 
\bibitem{omar} B. Davies, O. Foda, M. Jimbo, T. Miwa, and A. Nakayashiki, {\it Diagonalization of the $XXZ$ Hamiltonian by vertex operators}, 
Commun. Math. Phys. {\bf 151}, 89 (1993). 
\bibitem{dilsz} C. Dunning, M. Iba\~nez, J. Links, G. Sierra, and S.-Y. Zhao, {\it Exact solution of the $p + ip$
pairing Hamiltonian and a hierarchy of integrable models}, J. Stat. Mech. P08025 (2010). 
\bibitem{eks} V.Z. Enol'skii, V.B. Kuznetsov, and M. Salerno, {\it On the quantum inverse scattering method for the DST dimer},
Phys. D {\bf 68}, 138 (1993).
\bibitem{esks} V.Z. Enol'skii, M. Salerno, N.A. Kostov, and A.C. Scott, {\it Alternate quantizations of the discrete self-trapping
dimer}, Phys. Scripta {\bf 43}, 229 (1991).
\bibitem{esse} V.Z. Enol'skii, M. Salerno, A.C. Scott, and J.C. Eilbeck, {\it There's more than one way to skin Schr\"odinger’s cat},
Phys. D {\bf 59}, 1 (1992).
\bibitem{erbe} B. Erbe and J. Schliemann, {\it Different types of integrability and their relation to decoherence in central spin models}, Phys. Rev. Lett. {\bf 105}, 177602 (2010).
\bibitem{ef} F.H.L. Essler and M. Fagotti, {\it Quench dynamics and relaxation in isolated integrable quantum spin chains}
J. Stat. Mech.  064002 (2016).
\bibitem{ev} P. Etingof and A. Varchenko, {\it Solutions of the quantum dynamical Yang–Baxter equation and dynamical quantum groups},
Commun. Math. Phys, {\bf 196}, 591 (1998). 
\bibitem{holger} H. Frahm, J.H. Grelik, A. Seel, and T. Wirth, {\it Functional Bethe ansatz methods for the open $XXX$
chain}, J. Phys. A: Math. Theor. {\bf 44}, 015001 (2011).
\bibitem{fp} R. Franzosi and V. Penna, {\it Chaotic behavior, collective modes, and self-trapping in the dynamics of three coupled Bose-Einstein condensates}, Phys. Rev. E {\bf 67}, 046227 (2003).
\bibitem{maillet} L. Freidel and J.-M. Maillet, {\it Quadratic algebras and integrable systems}, Phys. Lett. B {\bf 262}, 278 (1991).  
\bibitem{g}  M. Gaudin, {\it Diagonalisation d'une classe d'Hamiltoniens de spin}, J. Phys. (Paris) {\bf 37}, The spin-$s$ homogeneous central spin model: exact spectrum and dynamics1087 (1976).
\bibitem{gldb} K.-L. Guan, Z.-M. Li, C. Dunning, and M.T. Batchelor
{\it The asymmetric quantum Rabi model and generalised P\"oschl-Teller potentials},   
J. Phys. A: Math. Theor. {\bf 51}, 315204 (2018).
\bibitem{jurco} B. Jurco, {\it Classical Yang-Baxter equations and quantum integrable systems}, J. Math. Phys. {\bf 30}, 1289 (1989)
\bibitem{kitaev} A. Kitaev, {\it Anyons in an exactly solved model and beyond}, Ann. Phys. {\bf 321},  2 (2006).
\bibitem{l} J. Larson, {\it Integrability versus quantum thermalization},  J. Phys. B: At. Mol. Opt. Phys. {\bf 46}, 224016 (2013).
\bibitem{ll} E. H. Lieb and W. Liniger, {\it Exact analysis of an interacting Bose gas. I. The general solution and the
ground state}, Phys. Rev. {\bf 130}, 1605 (1963).
\bibitem{lsm} E. Lieb, T. Schultz, and  D. Mattis, {\it Two soluble models of an antiferromagnetic chain},
Ann. Phys. {\bf 16}, 407 (1961).
\bibitem{lw} E.H. Lieb and F.Y. Wu, {\it Absence of Mott transition in an exact solution of the short-range, one-band model in one dimension } Phys. Rev. Lett. {\bf 20}, 1445 (1968).
\bibitem{links} J. Links, {\it Solution of the classical Yang--Baxter equation with an exotic symmetry, and 
integrability of a multi-species boson  tunnelling model}, Nucl. Phys. B {\bf 916}, 117 (2017).
\bibitem{lf} J. Links and A. Foerster, {\it On the construction of integrable closed chains with quantum supersymmetry}, 
J. Phys A: Math. Gen. {\bf 30}, 2483 (1997).
\bibitem{lh} J. Links and K.E. Hibberd, {\it Bethe ansatz solutions of the Bose-Hubbard dimer},
SIGMA  {\bf 2}, 095 (2006).
\bibitem{lzmg} J. Links,  H.-Q. Zhou, R.H. McKenzie,  and M.D. Gould, {\it Ladder operator for the one-dimensional Hubbard model}, Phys. Rev. Lett. {\bf 86},  5096 (2001).
\bibitem{lil} I. Lukyanenko, P. S. Isaac, and J. Links, {\it An integrable case of the $p + ip$ pairing 
Hamiltonian interacting with its environment}, J. Phys. A: Math. Theor. {\bf 49}, 084001 (2016).
\bibitem{mr} M.J. Martins and P.B. Ramos, {\it The quantum inverse scattering method for Hubbard-like models}, Nucl. Phys. B {\bf 522}, 413 (1998).
\bibitem{jim} J.B. McGuire, {\it Study of exactly soluble one-dimensional $N$-body problems}, J. Math. Phys. {\bf 5}, 622 (1964).
\bibitem{moroz} A. Moroz, {\it On the spectrum of a class of quantum models}, Europhys. Lett.  {\bf 100}, 60010 (2012).
\bibitem{ng} R.I. Nepomechie and X.-W. Guan, {\it The spin-$s$ homogeneous central spin model: exact spectrum and dynamics}, 
J. Stat. Mech. 103104 (2018). 
\bibitem{nw} R.I. Nepomechie and C. Wang, {\it Twisting singular solutions of Bethe’s equations}, J. Phys. A: Math. Theor. {\bf 47}, 505004 (2014).
\bibitem{ol} N. Oelkers and J. Links, {\it Ground-state properties of the attractive one-dimensional Bose-Hubbard model},  Phys. Rev. B {\bf 75}, 115119 (2007).
\bibitem{resh} N. Yu Reshetikhin, {\it A method of functional equations in the theory of exactly solvable quantum systems}, Lett. Math. Phys. {\bf 7}, 205 (1983).
\bibitem{sergeev} S. Sergeev, {\it Integrability of $q$-oscillator lattice model}, Phys. Lett. A {\bf 357}, 417 (2006).
\bibitem{shastry} B.S. Shastry, {\it Exact integrability of the 
one-dimensional Hubbard model}, Phys. Rev. Lett. {\bf 56}, 2453
(1986).
\bibitem{s} E.K. Sklyanin, {\it Separation of variables in the Gaudin model}, J. Soviet Math. {\bf 47}, 2473 (1989).
\bibitem{sklyanin} E.K. Sklyanin, {\it Boundary conditions for integrable quantum systems}, J. Phys. A: Math. Gen. {\bf 21}, 2375 (1988).
\bibitem{skrypnyk07} T. Skrypnyk, {\it Generalized Gaudin spin chains, nonskew symmetric $r$-matrices, and
reflection equation algebras}, J. Math. Phys. {\bf 48}, 113521 (2007).
\bibitem{skrypnyk09} T. Skrypnyk, {\it Non-skew-symmetric classical $r$-matrices and integrable cases of the 
reduced BCS model}, J. Phys. A: Math. Theor. {\bf 42}, 472004 (2009).
\bibitem{taka} M. Takahashi, {\it One-dimensional Heisenberg model at finite temperature}, Prog. Theor. Phys. {\bf 46}, 401 (1971).
\bibitem{tf} L.A. Takhtadzhan and L.D. Faddeev, {\it The quantum method of the inverse problem and the Heisenberg $XYZ$ model}, 
Russ. Math. Surv. {\bf 34}, 11 (1979).
\bibitem{torr} A. Torrielli, {\it Classical integrability},  J. Phys. A: Math. Theor. {\bf 49}, 323001
(2016).
\bibitem{vp} J. von Delft and R. Poghossian,  {\it Algebraic Bethe ansatz for a discrete-state BCS pairing model}, Phys. Rev. B {\bf 66}, 134502 (2002).
\bibitem{wz} P.B. Wiegmann and A.V. Zabrodin, {\it Bethe-Ansatz for the Bloch electron in magnetic field}, Phys. Rev. Lett. {\bf 72}, 1890 (1994). 
\bibitem{yang} C.N. Yang, {\it Some exact results for the many-body problem in one dimension with delta-function interaction}, Phys. Rev. Lett.  {\bf 19}, 1312 (1967).
\bibitem{ytfl} L.H. Ymai, A.P. Tonel, A. Foerster,  and J. Links,  
{\it Quantum integrable multi-well tunneling models}, J. Phys A: Math. Theor. {\bf 50},  264001, (2017).
\bibitem{usw} Y. Umeno, M. Shiroishi, and M. Wadati, 
{\it Fermionic R-operator and integrability of the one-dimensional Hubbard model}, J. Phys. Soc. Jpn. {\bf 67}, 2242 (1998).
\bibitem{yu} Y. Yu, {\it Gauge symmetry in Kitaev-type spin models and index theorems on odd manifolds}, Nucl. Phys. B {\bf 799}, 345 (2008).  
\bibitem{yao} Y.-Z. Zhang, {\it On the solvability of the quantum Rabi model and its 2-photon and two-mode generalizations},  J. Math. Phys. {\bf 54}, 102104 (2013).
\bibitem{zlmg} H.-Q. Zhou, J. Links, R.H. McKenzie, and M.D. Gould,  {\it Superconducting correlations in metallic nanoparticles: Exact solution of the BCS model
by the algebraic Bethe ansatz}, Phys. Rev. B {\bf 65}, 060505(R) (2002).





\end{thebibliography}
\end{document}